\documentclass[11pt]{article}

\usepackage{amsthm,amsmath,amssymb,graphicx}
\usepackage[english]{babel}

\newcommand{\iin}{\rm{in}}

\newcommand{\R}{{\mathbb R}}
\newcommand{\PP}{{\mathbb P}}
\newtheorem{defn}{Definition}
\newtheorem{cor}{Corollary}
\newtheorem{thm}{Theorem}

\renewenvironment{proof}{\noindent{\bf Proof:} }{\hfill $\square$ \\}

\begin{document}
\setlength{\baselineskip}{1.3 \baselineskip}

\title{Discrete multi-colour random mosaics with an application to network extraction}
\author{M.N.M.\ van Lieshout \\ \mbox{} 
\\ CWI \\ Amsterdam, The Netherlands}
\maketitle
\begin{verse}
{\footnotesize
\noindent
We introduce a class of random fields that can be understood as discrete versions of 
multi-colour polygonal fields built on regular linear tessellations. We focus first 
on consistent polygonal fields, for which we show Markovianity and solvability by 
means of a dynamic representation. This representation forms the basis for new 
sampling techniques for Gibbsian modifications of such fields, a class which covers
lattice based random fields. A flux based modification is applied to the extraction 
of the field tracks network from a SAR image of a rural area.
}
\end{verse}

\large{
\begin{center}
In memory of Tomasz F.\ Schreiber.
\end{center}
}

\section{Introduction} 

In the 1980s, Arak and Surgailis introduced a class of planar Markov 
fields whose realisations form a coloured tessellation of the plane. 
The basic idea is to use the lines of an isotropic Poisson line 
process as a skeleton on which to draw polygonal contours with the 
restriction that each line cannot be used more than once. Note that 
many tessellations can be drawn on the same skeleton and that the 
contours may be nested. The polygons are then coloured randomly 
subject to the constraint that adjacent ones must have different colours. 
Formally, the probability distribution of such polygonal Markov fields 
is defined in terms of a Hamiltonian with respect to the law of the
underlying Poisson line process, which can be chosen in such a way 
that many of the basic properties of the Poisson line process (including 
consistency, Poisson line transects, and an explicit expression for its 
probability distribution on the hitting set of bounded domains) carry 
over and a Markov property holds.  Another useful feature is that 
a dynamic representation in terms of a particle system is available. See 
\cite{ArakSurg89,ArakSurg91} for further details and \cite{ArakClifSurg93}
or \cite{Thal11} for alternative point rather than line based models.
The special case where all vertices have degree three was studied in
\cite{MackMile02}.

The simplest and most widely studied example of a planar Markov field 
is the Arak process \cite{Arak82} which consists of self-avoiding 
closed polygonal contours. Hence interaction is restricted to a hard
core condition on the contours, and there are exactly two colourings 
using the labels `black' and `white' such that adjacent polygons 
have different colours. For the Arak process, the Hamiltonian is 
proportional to the total contour length by a factor two. One may 
introduce further length-interaction by changing the proportionality 
constant. Doing so, Nicholls \cite{Nich01} and Schreiber \cite{Schr06} 
consider the separation between the black and white regions as the 
interaction gets stronger; Van Lieshout and Schreiber \cite{LiesSchr07} 
develop perfect simulation algorithms for these models. For the more 
general case in which both length and area terms feature in the 
Hamiltonian, Schreiber \cite{Schr05} proposes a Metropolis--Hastings
scheme based on the dynamic representation of the Arak process, which 
Kluszczy\'nski et al. \cite{KlusLiesSchr05} adapt and implement to 
solve foreground/background image segmentation problems. An anisotropic 
Arak process can be defined through local activity functions instead of
the length functional \cite{Schr08,Schr10} and allows for increased
flexibility while preserving desirable basic properties including a
dynamic representation. This representation forms the basis of a 
stochastic optimisation algorithm in the context of image segmentation, 
implemented by Matuszak and Schreiber \cite{MatuSchr12}, and helps to
gain insight into the higher order correlation structure \cite{Schr08}.

Polygonal Markov field models with contours that may also be joined by a 
vertex of degree three or four \cite{ArakSurg89,ArakSurg91} are much less 
well-understood due to the more complicated interaction structure. Papers 
in this direction include \cite{ClifNich94,KlusLiesSchr07,PaskThru05}.  

In a previous paper \cite{SchrLies10}, we introduced a class of binary
random fields that can be understood as discrete counterparts 
of the two-colour Arak process. The aim of the present paper is to 
extend this construction to allow for an arbitrary number of colours
and to relax the assumption in \cite{SchrLies10} that no polygons of the
same colour can be joined by corners only. Our construction is two-staged: 
first a collection of lines is fixed to serve as a skeleton for drawing
polygonal contours (a regular lattice being the generic example), then 
the resulting polygons are coloured in such a way that adjacent ones do
not have the same colour. The analogy with continuum polygonal Markov 
fields is exploited to define Hamiltonians that are such that the desirable 
properties of these processes mentioned above hold. Moreover, we propose 
new simulations techniques that combine global changes with the usual 
local update methods employed for random fields on finite graphs 
\cite{Wink03}. It should be stressed, though, that the discrete models
considered in this paper are not versions of continuum polygonal Markov
fields conditioned on having their edges fall along a given finite 
collection of lines.

The plan of this paper is as follows. In Section~\ref{S:definition} we 
define a family of admissible multi-colour polygonal configurations built 
on regular linear tessellations, and define discrete polygonal fields 
with special attention to consistent ones. In Section~\ref{S:dynamics} 
we present a dynamic representation of consistent polygonal fields, which 
is used to prove the main properties of such models. In Section~\ref{S:sampler}, 
we exploit the dynamic representation to develop a simulation method for 
consistent polygonal fields. The method is generalised to arbitrary Gibbs 
fields with polygonal realisations and applied to the detection of linear 
networks in images in Section~\ref{S:application}. We conclude with a 
discussion.

\section{Random fields with polygonal realisations}
\label{S:definition}

First, recall the definition of a regular linear tessellation from 
\cite{SchrLies10}.

\begin{defn}
\label{D:T}
A {\rm regular linear tessellation} of the plane is a countable family 
${\cal T}$ of straight lines in ${\R}^2$ such that no three lines of 
${\cal T}$ intersect at one point and such that any bounded subset of 
the plane is hit by at most a finite number of lines from ${\cal T}$.
\end{defn}

For a bounded open convex set $D$, ${\cal T}$ induces a partition of $D$ 
into a finite collection $D_{\cal T}$ of convex cells of polygonal shapes, 
possibly chopped off by the boundary. Below we shall always assume that 
the boundary $\partial D$ of $D$ contains no intersection points of lines 
from ${\cal T}$ and that the intersection of each line $l \in {\cal T}$ 
with $\partial D$ consists of exactly two points. To each line $l$, we 
ascribe a fixed activity parameter $\pi_l \in (0,1)$ to allow for 
the possibility to favour some lines over others. 

The next step is to assign a colour to each of the convex cells in
the partition of $D$ induced by the lines in ${\cal T}$. Write
$\{ 1, \dots, k \}$ for the set of colour labels. In this paper, we 
concentrate on the case that $k > 2$. Such a colouring gives rise to a 
graph whose edges are formed by the boundaries between cells that have 
been assigned different colours. For technical convenience, we shall 
assume that edges are open, that is, they do not contain the vertices 
in which they intersect. Faces of the graph, which are unions of cells 
of $D_{\cal T}$, are said to be adjacent if they share a common edge. 

\begin{defn}
The family $\hat \Gamma_D({\cal T})$ of {\rm admissible coloured
polygonal configurations} in $D$ built on ${\cal T}$ consists of all 
coloured planar graphs $\hat \gamma$ in the topological closure 
$\overline{D} = D \cup \partial D$ of $D$  such that
\begin{itemize}
\item all edges lie on the lines of ${\cal T}$;
\item all interior vertices,  i.e. those lying in $D$,
      are of degree $2$, $3$ or $4$; 	
\item all boundary vertices, i.e. those lying on $\partial D$, 
      are of degree $1$; 
\item no adjacent faces share the same colour.
\end{itemize}
\end{defn}

Throughout this paper, the notation $\gamma$ is used for (admissible) planar 
graphs and the hat notation $\hat \gamma$ for the graph with colours assigned 
to its faces. In this notation, $\Gamma_D({\cal T})$ stands for the family
of all planar graphs $\gamma$ arising as interfaces between differently coloured 
faces in $\hat{\gamma} \in \hat{\Gamma}_D({\cal T})$. Note that for the 
case $k=2$ treated in \cite{SchrLies10}, all interior vertices have degree two. 
Vertices of degree two are also known as V-vertices, those of degree three 
as T-vertices, and vertices of degree four as X-vertices.

To avoid confusion, it is important to distinguish between $D_{\cal T}$ and
the members of $\Gamma_D({\cal T})$, even though they all partition $D$. 
In the sequel we shall reserve the terms `segments' and `nodes' for the former,
'edges' and 'vertices' for the latter. Thus, nodes are intersection points of 
lines in ${\cal{T}}$, which are joined by segments. Edges of $\gamma$ are maximal 
unions of connected collinear segments that are not broken by other such
segments lying on the graph (corresponding to vertices of degree
three or four). Primary edges are maximal unions of connected collinear 
segments, possibly consisting of multiple edges due to T- or X-vertices.
Likewise, a vertex of $\gamma$ is a point where two edges of $\gamma$ meet or 
where an edge of $\gamma$ meets the boundary $\partial D$, nodes lying on the 
interior of graph edges are not considered to be vertices. These concepts are
illustrated in Figure~\ref{F:graphs}. The family ${\cal T}$ consists of two
orthogonal line bundles and induces diamond shaped cells. One element
of $\Gamma_D({\cal T})$ is plotted in Figure~\ref{F:graphs}. Consider the
polygonal face indicated by `C' that is chopped off by the boundary. There are 
$2$ boundary vertices and $12$ interior vertices, $5$ of degree $2$, $7$ of 
degree $3$ and none of degree $4$. We count $17$ nodes on $16$ complete segments, 
with $2$ segments partly visible due to truncation by the boundary. $C$ has $13$ 
edges and its boundary lies on $8$ primary edges.
Note that in the literature, primary edges are sometimes also known as sides.  
See e.g.\ \cite{WeisCowa11} for a discussion on nomenclature for tessellations. 

\begin{figure}
\begin{center}
  \fbox{\includegraphics[height=40mm]{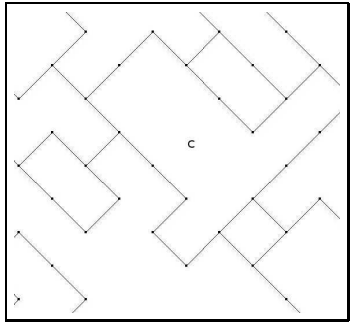}}
\end{center}
\caption{Admissible polygonal configuration.}
\label{F:graphs}
\end{figure}

We are now ready to define (discrete) polygonal fields. 

\begin{defn} Let ${\cal T}$ be a regular linear tessellation and ascribe fixed 
activity parameters $\pi_l \in (0,1)$ to the lines $l \in {\cal T}$. For a function 
${\cal H}_D : \hat{\Gamma}_D({\cal T}) \mapsto {\mathbb R} \cup \{ + \infty \}$,
the {\rm (discrete) polygonal field} $\hat{\cal A}_{{\cal H}_D}$ with Hamiltonian 
${\cal H}_D$ is the random element in $\hat{\Gamma}_D({\cal T})$ such that
\begin{equation}
\label{PolyField}
{\mathbb P}\left(\hat{\cal A}_{{\cal H}_D} = \hat{\gamma} \right) =
  \frac{\exp(-{\cal H}_D(\hat{\gamma})) \prod_{e \in E^*(\gamma)} \pi_{l[e]}
}{
 {\cal Z}[{\cal H}_D]
}.
\end{equation}
Here $E^*(\gamma)$ denotes the set of primary edges in $\gamma$ and 
$l[e] \in {\cal T}$ is the straight line containing the open edge $e$.
\end{defn}

The constant
\begin{equation}
\label{PartitionFunction}
  {\cal Z}[{\cal H}_D] = \sum_{\hat{\theta} \in \hat{\Gamma}_D({\cal T})}
  \exp(-{\cal H}_D(\hat{\theta})) \prod_{e \in E^*(\theta)} \pi_{l[e]}
\end{equation}
that ensures that ${\mathbb{P}}$ is a probability distribution is called
the partition function. Note that the polygonal field is a Gibbs field 
with Hamiltonian 
\[
{\cal H}_D(\hat{\gamma}) - \sum_{e \in E^*(\gamma)} \log \pi_{l[e]}.
\]
The terms $\log \pi_{l[e]}$ represent the energy needed to create the
edges. We prefer to use the term polygonal field since the consistent
fields to be considered shortly are inspired by the Arak--Surgailis 
polygonal Markov fields in the continuum, and, more importantly, 
because the graph-theoretical formulation leads us to define novel simulation
techniques for discrete random fields.

A careful choice of Hamiltonian in (\ref{PolyField}) leads to consistent 
polygonal fields. Recall that $k$ is the number of colour labels.

\begin{defn}
\label{D:consistent}
Let $\alpha_V \in [0,1]$ and set
\[
\alpha_X  =  1 - \alpha_V; 
\quad \quad
\alpha_T  =  \frac{1}{2} \left( 1 - \frac{k-2}{k-1} \alpha_X \right); 
\quad \quad
\epsilon  =  \frac{\alpha_V}{k-1} + \frac{k-2}{k-1} \alpha_T.
\]
The polygonal fields (\ref{PolyField}) defined by Hamiltonians
of the form
\begin{eqnarray}
\nonumber
{\Phi}_D(\hat\gamma) & = & - N_V(\gamma) \log \alpha_V 
  - N_T(\gamma) \log( (k-1) \alpha_T )
  - N_X(\gamma) \log( (k-1) \alpha_X ) \\
\nonumber
& + & {\rm{card}}(E(\gamma)) \log(k-1) \\
\label{E:consistent}
& - & \sum_{e \in E(\gamma)} \sum_{l \in {\cal T},\; l \nsim e} 
  \log(1- \epsilon \pi_{l}) 
+ \sum_{n(l_1,l_2) \in \gamma} 
  \log\left(1- \frac{\alpha_V}{k-1} \pi_{l_1}\pi_{l_2}\right)
\end{eqnarray}
are referred to as {\rm consistent polygonal fields}. Here,
$N_V$, $N_T$ and $N_X$ denote the number of V-, T- and X-vertices,
${\rm{card}}(E(\gamma))$ is the number of edges in $\gamma$, 
$n(l_1,l_2) \in \gamma$ ranges through the nodes of ${\cal{T}}$
that either lie on edges of $\gamma$ or coincide with one of its
vertices, and $l \nsim e$ means that the line $l$ intersects but is
not collinear with $e$. We use the convention that $0 \times \infty
= 0$.
\end{defn}

A few remarks are in order. The parameter $\alpha_V$  controls the 
relative frequency of $V$-vertices, $\alpha_T$ and $\alpha_X$ that
of vertices of degrees three and four respectively. These parameters are not
independent and, given the number of colours $k$, $\alpha_T$ and $\alpha_X$
are uniquely determined by $\alpha_V$. This dependence will become more
explicit in the dynamic representation to be derived in 
Section~\ref{S:dynamics}. Typical realisations for $k=3$ and $\alpha_V$ 
equal to $0$, $1/2$ and $1$ are shown in Figure~\ref{F:Arak}. In the left-most
panel, $\alpha_V = 0$ and there no vertices having degree two; in the
right-most panel, $\alpha_X = 0$ so that no vertices have degree four.
The central panel displays a coloured configuration with vertices of all 
degrees. Visually, the three patterns are strikingly different.

\begin{figure}
\begin{center}
  \fbox{\includegraphics[height=40mm]{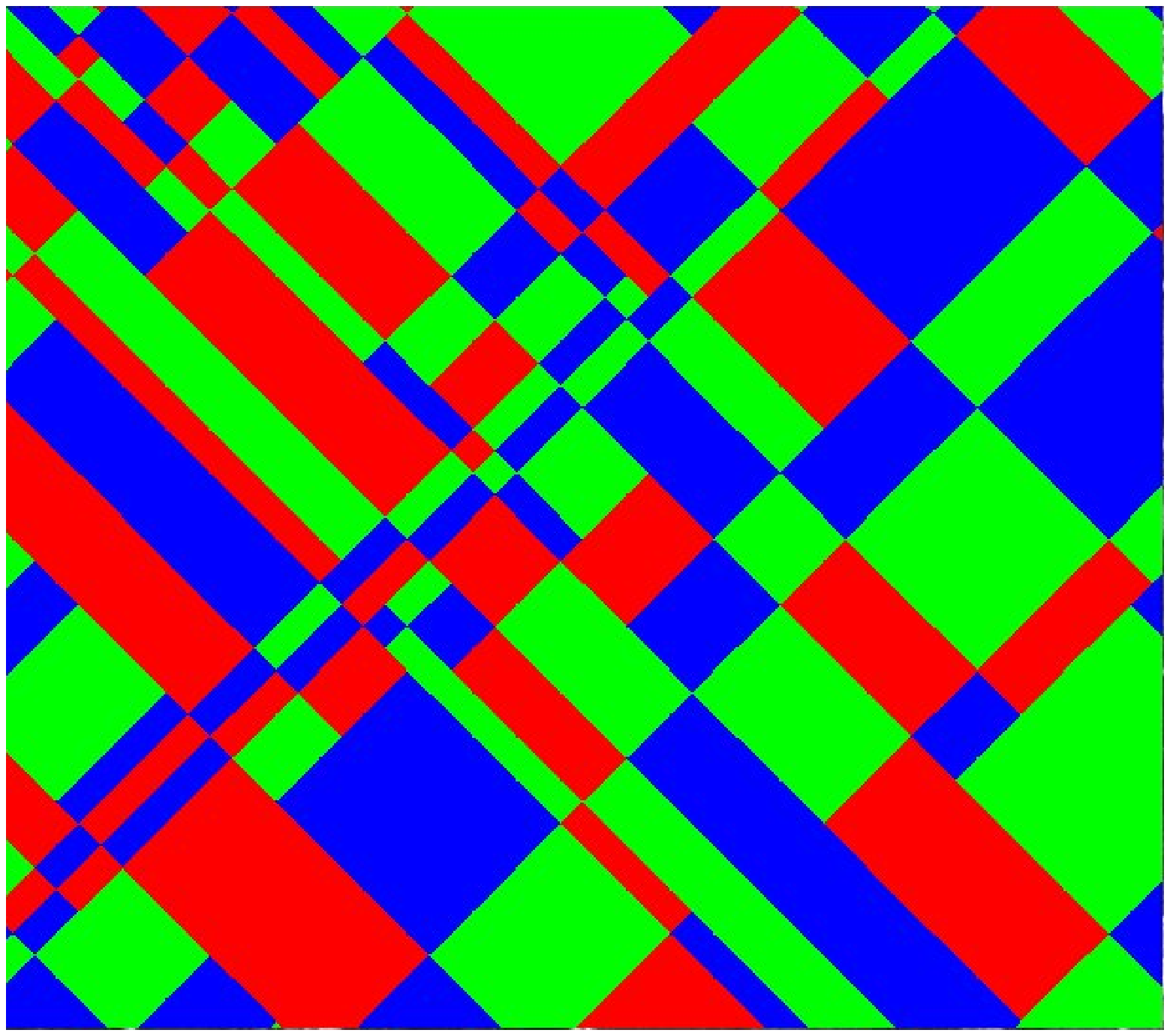}}
  \fbox{\includegraphics[height=40mm]{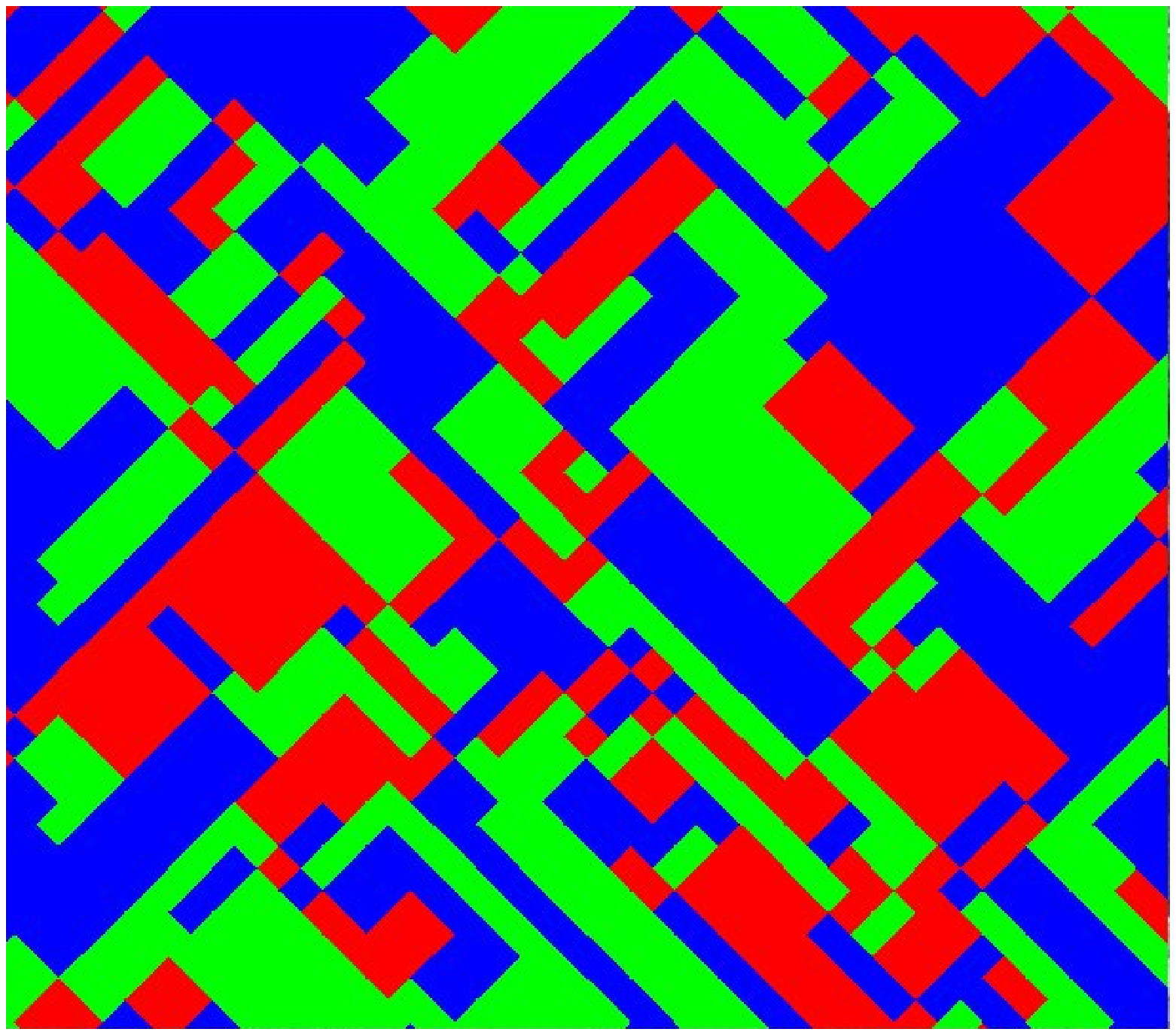}}
  \fbox{\includegraphics[height=40mm]{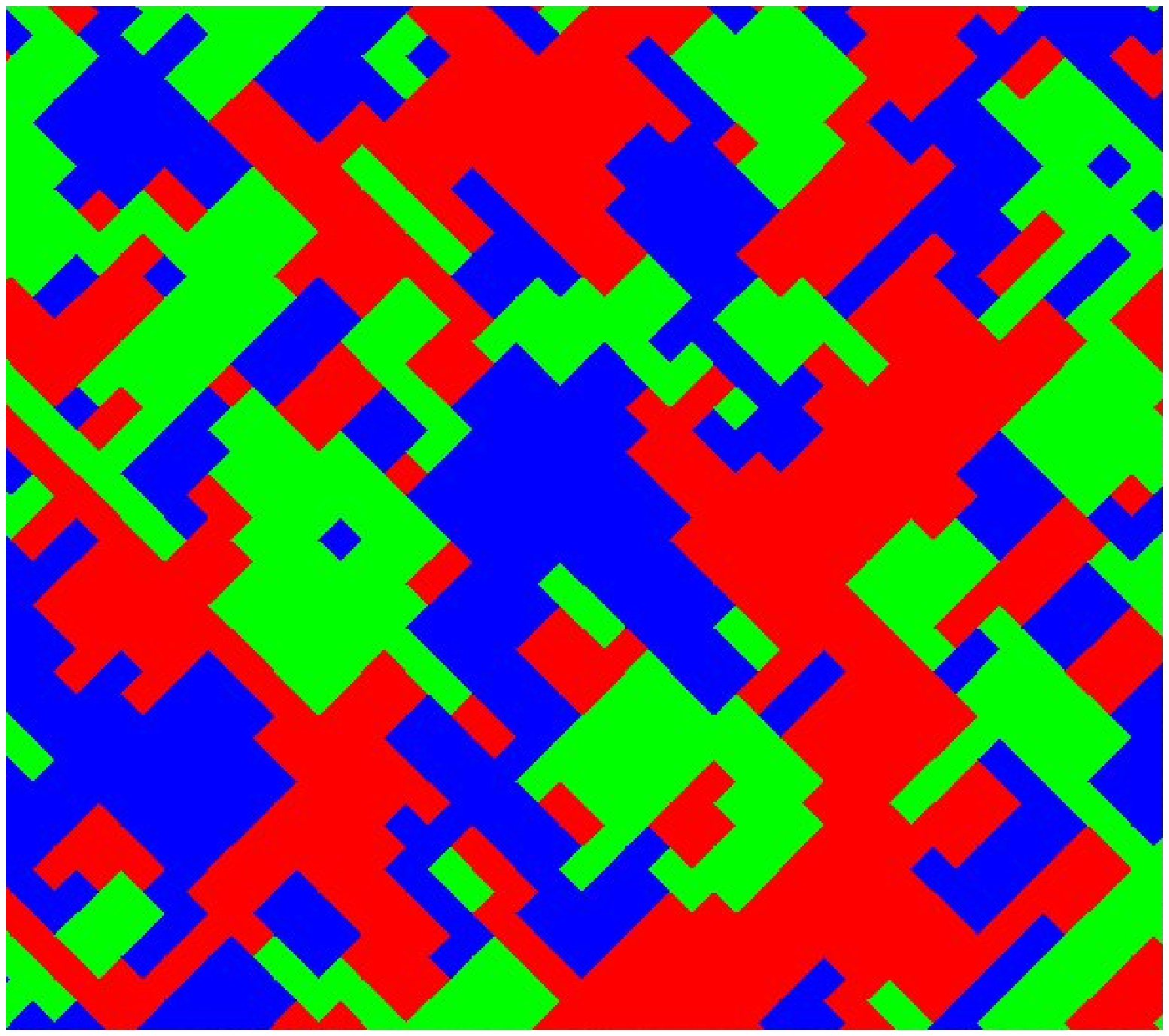}}
\end{center}
\caption{Realisations of $\hat{\cal{A}}_{\Phi_D}$ with $\alpha_V = 0$ (left),
$\alpha_V = 1/2$ (centre) and $\alpha_V = 1$ (right). In all panels,
the number of colour labels $k=3$.}
\label{F:Arak}
\end{figure}

For now, let us consider the special case $k=2$ and $\alpha_V=1$. First note 
that the family $\hat\Gamma_D({\cal{T}})$ of admissible coloured polygonal 
configurations does not include any member with an interior vertex of degree 
three. Therefore $N_T \equiv 0$ almost surely. Moreover, as $\alpha_X = 
1 - \alpha_V = 0$, the probability of any $\hat \gamma$ that contains 
X-vertices is zero. Hence, almost surely all interior vertices have degree 
two and $E^*(\gamma) = E(\gamma)$. Moreover, $\log(k-1) = \log \alpha_V = 0$,
and (\ref{E:consistent}) simplifies to
\[
 {\Phi}_D(\hat\gamma) = 
 -  \sum_{e \in E(\gamma)} \sum_{l \in {\cal T},\; l \nsim e} 
  \log(1- \pi_{l}) 
+ \sum_{n(l_1,l_2) \in \gamma} 
  \log\left(1- \pi_{l_1}\pi_{l_2}\right),
\]
cf.\ \cite{SchrLies10}.

The nomenclature in Definition~\ref{D:consistent} is justified by the following result.

\begin{thm}
\label{T:consistency}
The polygonal field $\hat{\cal{A}}_{\Phi_D}$ is consistent:
For bounded open convex $D' \subseteq D \subseteq {\mathbb R}^2$,
the field $\hat{\cal A}_{\Phi_{D}} \cap D'$ coincides in distribution
with $\hat{\cal A}_{\Phi_{D'}}.$
\end{thm}

By letting $D$ increase to ${\mathbb R}^2$, Kolmogorov's theorem allows us 
to construct the whole plane extension of the process $\hat{\cal A}_{\Phi}$ 
such that the distribution of $\hat{\cal A}_{\Phi_D}$ coincides with that of 
$\hat{\cal A}_{\Phi} \cap D$ for all bounded open convex $D \subseteq 
{\mathbb R}^2.$ The proof of Theorem~\ref{T:consistency} relies on a dynamic
representation, which is the topic of the next section.

\section{Dynamic representation of consistent polygonal 
fields}
\label{S:dynamics}

Below, we present a dynamic representation for discrete consistent 
polygonal fields in analogy with the corresponding representation in 
\cite[Sections 4 and 5]{ArakSurg89}. The idea underlying this 
construction is to represent the edges of the polygonal field as
the trajectory of a one-dimensional particle system evolving in time.
More specifically, we interpret $D$ as a set of time-space points 
$(t,y) \in D$ and refer to $t$ as the time coordinate, to $y$ 
as the spatial coordinate of a particle. In this language, a straight 
line segment in $D$ stands for a piece of the time-space trajectory of 
a moving particle. Compared to the bi-coloured case discussed in 
\cite{SchrLies10}, note that in the current multi-colour context, in 
general it is not sufficient to consider colourless graphs only and
assign colourings with equal probability afterwards. 
Moreover, the interaction structure is not restricted to a hard core 
condition.

For convenience, we assume that no line in ${\cal{T}}$ or segment of 
$\partial D$ is parallel to the spatial axis. Since we might simply
rotate the coordinate system otherwise, these assumptions do not lead to a 
loss of generality. Recall that by assumption, each line $l$ of ${\cal T}$ 
intersects the boundary at two points. The two intersection points are 
ordered according to time and the one with the smaller time coordinate is 
denoted ${\iin} (l, D)$. Furthermore, no three lines of a regular linear 
tessellation intersect in a single point.

To define the dynamics, the left-most point of $\bar D$ is assigned a random
colour chosen uniformly from the $k$ possibilities. Let particles be born 
independently of other particles
\begin{itemize}
\item with probability $\alpha_V \pi_{l_1} \pi_{l_2} / (k-1)$ at each node 
      $n(l_1,l_2)$, that is, the intersection of two lines $l_1$ and $l_2$ in 
      ${\cal T}$, which falls in $D$  (interior birth site),
\item with probability ${\pi_l} / {(1+\pi_l)}$ at each entry point 
      ${\iin} (l,D)$ of a line $l \in {\cal T}$ into $D$
      (boundary birth site).
\end{itemize}
Each interior birth site $n(l_1,l_2)$ emits two particles moving with 
initial velocities such that the initial segments of their trajectories 
lie on the lines $l_1$ and  $l_2$ of the tessellation that emanate from the 
birth site, unless another particle (either a single one or two colliding particles)
previously born hits the site, in which case the birth does not occur. 
Each boundary birth site ${\iin} (l,D)$ emits a single particle moving with 
initial velocity such that the initial segment of its trajectory lies on $l$.
Note that no precaution similar to the one for interior birth sites above is 
needed because boundary birth sites cannot be hit by previously born particles. 
The initial trajectory or trajectories of a birth event bound a new
polygonal region, the colour of which is chosen uniformly from the $k-1$
colours that differ from that of the polygon just prior to the birth, or
in other words, lying to the left of the birth site.

\begin{figure}[thb]
\begin{center}
{\includegraphics[height=2.5cm,width=2.5cm]{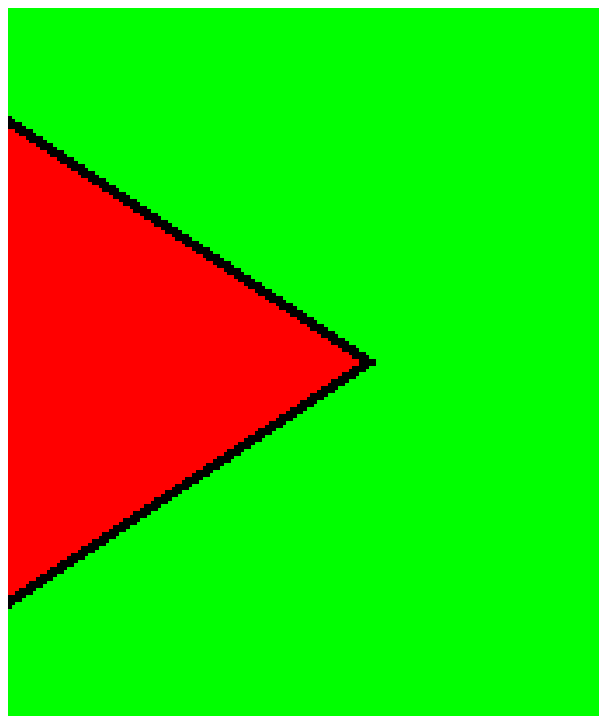}}
\hspace{1cm}
{\includegraphics[height=2.5cm,width=2.5cm]{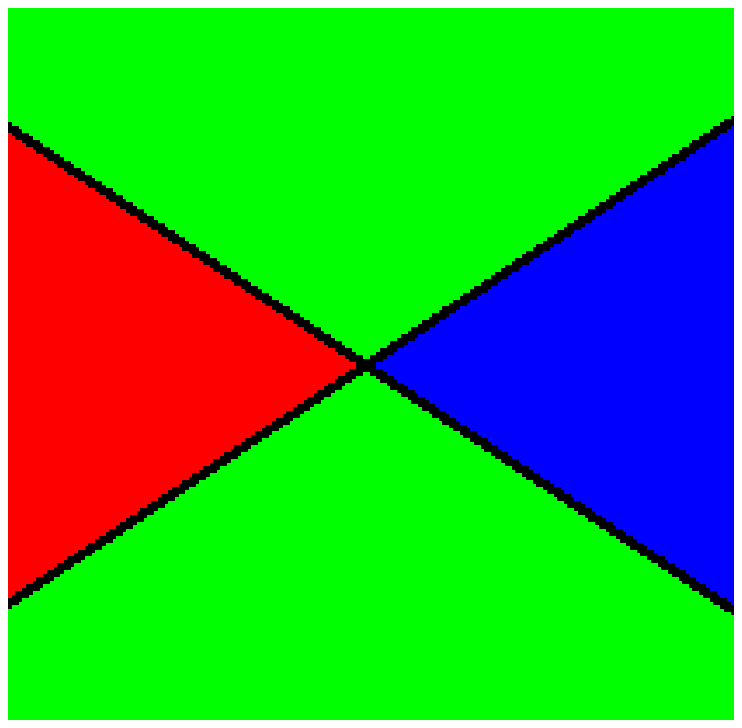}}
\end{center}
\begin{center}
{\includegraphics[height=2.5cm,width=2.5cm]{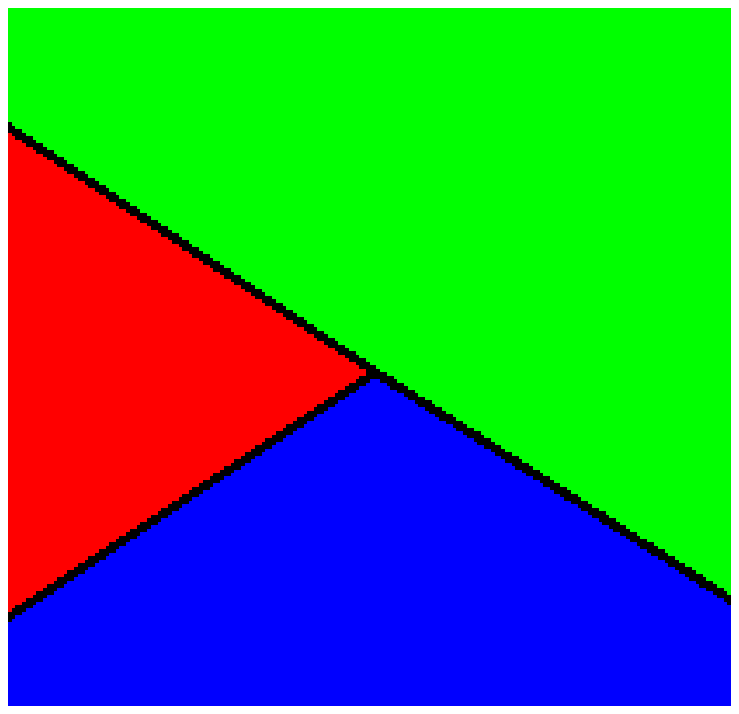}}
\hspace{1cm}
{\includegraphics[height=2.5cm,width=2.5cm]{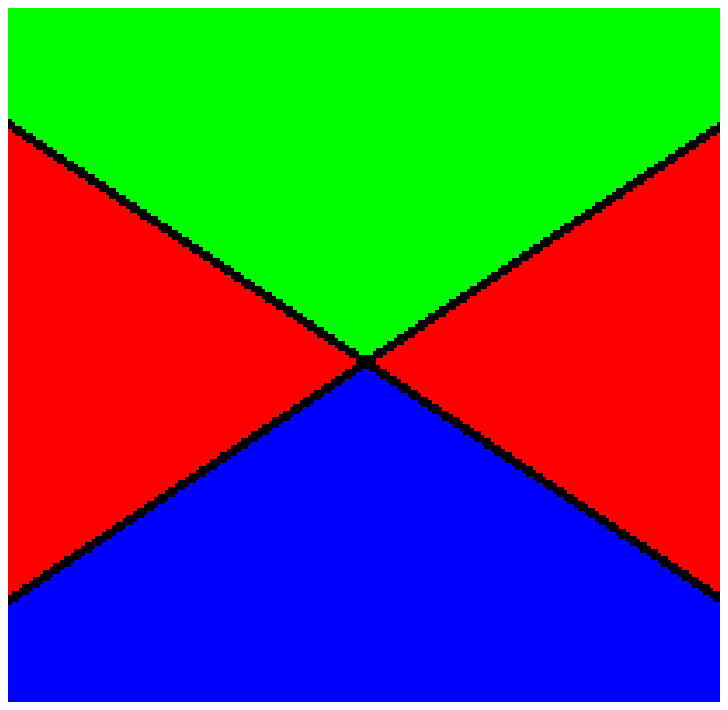}}
\end{center}
\caption{Top row ({\bf{(E3) [a]}}): In the left-most panel, both particles die;
they survive in the panel on the right and create a new polygon coloured blue.
Bottom row ({\bf{(eE3) [b]}}): In the left-most panel, one particles survives;
both particles survive in the panel on the right and create a new polygon 
coloured red.}
\label{F:E3}
\end{figure}

All particles evolve independently in time according to the following rules.
\begin{description}
\item{\bf (E1)} Between the critical moments listed below each particle
                  moves freely with constant velocity.
\item{\bf (E2)} When a particle hits the boundary $\partial D,$ it dies.
\item{\bf (E3)} In case of a collision of two particles, that is, equal spatial 
        coordinates $y$ at some time $t$ with $(t,y) = n(l_i, l_j) \in D$, 
        distinguish the following cases.
\begin{description}
\item[a] If the colours above and below $(t,y)$ are identical, say $i$, 
         with probability $\alpha_V$ both particles die. With probability 
         $\alpha_X$ both particles survive to create a new polygon whose colour 
         is chosen uniformly from those not equal to $i$ (cf.\ Figure~\ref{F:E3}).
\item[b] If the colours above and below $(t,y)$ are different, say $i$ and $j$,
         with probability $\alpha_T$, each of the two particles survives while
         the other dies. With probability $(k-2) \alpha_X/ (k-1) = 1 - 2 \alpha_T$, 
         both particles survive to create a new polygon whose colour is chosen
         uniformly from those not equal to either $i$ or $j$ (cf.\ Figure~\ref{F:E3}).
\end{description}
         Recall that a collision prevents a birth from happening at the node.
\item{\bf (E4)} Whenever a particle moving in time-space along $l_i \in 
        {\cal T}$ reaches a node $n(l_i,l_j),$ it changes its velocity so 
        as to move along $l_j$ with probability $ \alpha_V \pi_{l_j} / (k-1)$,
        it splits into two particles moving along $l_i$ and $l_j$ with 
        probability $(k-2) \alpha_T \pi_{l_j} / (k-1)$ and keeps moving along 
        $l_i$ otherwise (with probability $1-\epsilon \pi_{l_j}$). In case of a
        split, a new polygon is created whose colour is chosen uniformly from 
        the $k-2$ possibilities. See Figure~\ref{F:E4}.
\end{description}

The dynamics described above define a random coloured polygonal configuration
$\hat {\cal{D}}_D$. The key observation is that its distribution is identical
to that of $\hat {\cal{A}}_{\Phi_D}$.

\begin{thm}
\label{T:sameLaw}
The random elements $\hat {\cal{A}}_{\Phi_D}$ and $\hat {\cal{D}}_D$ coincide
in distribution.
\end{thm}

\begin{figure}[bht]
\begin{center}
{\includegraphics[height=2.5cm]{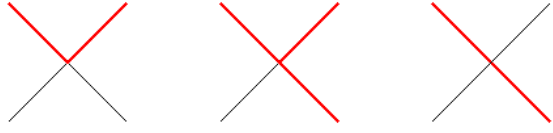}}
\end{center}
\caption{{\bf (E4):} Velocity update (left), split (middle) and continuation
(right) of a trajectory outlined in red.}
\label{F:E4}
\end{figure}

\begin{proof}
In order to calculate the probability that some $\hat\gamma \in \hat\Gamma_D({\cal T})$
is generated by the particle dynamics {\bf E1--E4}, observe that
\begin{itemize}
\item each edge $e \in E(\gamma)$ whose initial (lower time coordinate) 
  vertex lies on $\partial D$ yields a factor ${\pi_{l[e]}} / (1+\pi_{l[e]})$
  (boundary birth site) times $\prod_{l \in {\cal T},\; l \nsim e}
  (1-\epsilon \pi_{l})$ (no velocity/split updates along $e$) times $1/(k-1)$
  for the colour;
\item the two edges $e_1,e_2 \in E(\gamma)$ emanating from a common interior
  birth site $n(l_1,l_2)$ yield a factor
  $\alpha_V \pi_{l_1} \pi_{l_2}/(k-1)$ (the birth probability)
  times $\prod_{i=1}^2 \prod_{l \in {\cal T},\; l \nsim e_i} (1-\epsilon \pi_{l})$ 
  (no velocity/split updates along $e_i$) times $1/(k-1)$ for the colour;
\item each edge $e \in E(\gamma)$ arising due to a velocity update 
  yields a factor $\alpha_V \pi_{l[e]} / (k-1)$ (velocity update 
  probability) times  $\prod_{l \in {\cal T},\; l \nsim e} (1-\epsilon \pi_{l})$ 
 (no velocity/split updates along $e$);
\item the two edges $e_1,e_2 \in E(\gamma)$ arising from a split at node
  $n(l_1,l_2)$ where $e_1$ is a continuation and $e_2$ a new direction yield a 
 factor $(k-2) \alpha_T \pi_{l_2}/(k-1)$ (the split probability)
 times $\prod_{i=1}^2 \prod_{l \in {\cal T},\; l \nsim e_i} (1-\epsilon \pi_{l})$ 
 (no velocity/split updates along $e_i$) times $1/(k-2)$ for the colour;
\item the two edges $e_1, e_2 \in E(\gamma)$ arising due to an X-collision 
  contribute a factor $\alpha_X / (k-1)$ times 
  $\prod_{i=1}^2 \prod_{l \in {\cal T},\; l \nsim e_i} (1-\epsilon \pi_{l})$
  (no velocity/split updates along $e_i$); 
\item each edge $e$ emanating from a T-collision contributes a factor $\alpha_T$ 
  times $\prod_{l \in {\cal T},\; l \nsim e} (1-\epsilon \pi_{l})$ (no velocity/split 
  updates along $e$);
\item each collision of V-type contributes a factor $\alpha_V$;
\item the initial colour choice contributes a factor $1/k$;
\item the absence of birth sites in nodes $n(l_1,l_2) \in D$ that do not belong
  to $\gamma$ yields the factor $\prod_{n(l_1,l_2) \in D   \setminus \gamma} 
 (1- \alpha_V \pi_{l_1}\pi_{l_2}/(k-1))$ ;
\item the absence of boundary birth sites at those entry points into $D$ of 
  lines of ${\cal T}$ which do not give rise to an edge of $\gamma$ yields 
 the factor 
 $\prod_{l \in {\cal T},\; l \cap D \neq \emptyset,\;{\iin} (l,D) \not\in \gamma}
 (1+\pi_{l})^{-1}.$
\end{itemize}
Collecting all factors implies that the probability of $\hat\gamma$ is
\[
   \left( \prod_{e \in E(\gamma)} \prod_{l \in {\cal T},\; l \nsim e} 
   (1- \epsilon \pi_{l}) \right)
   \left( \prod_{n(l_1,l_2) \in D\setminus \gamma} 
   \left(1- \frac{\alpha_V}{k-1} \pi_{l_1}\pi_{l_2}\right) \right)
   \times \left( \prod_{l \in {\cal T},\; l \cap D \neq \emptyset} 
                \frac{1}{1+\pi_{l}} \right) 
\]
times
\[
\frac{1}{k} \,   
\alpha_V^{N_V} \, \left( \alpha_T (k-1) \right)^{N_T} \,
\left( \alpha_X (k-1) \right)^{N_X} 
\prod_{e \in E^*(\gamma)} \pi_{l[e]} \times 
\prod_{e \in E(\gamma)} \frac{1}{k-1},
\]
which completes the proof.
\end{proof}

As an immediate consequence of the proof of Theorem~\ref{T:sameLaw},
we obtain an explicit and simple expression for the partition function of 
consistent polygonal fields.

\begin{cor}
The partition function (\ref{PartitionFunction}) is given by
\[
k \prod_{l \in {\cal T},\; l \cap D \neq \emptyset} ( 1+\pi_{l} ) 
  \prod_{n(l_1,l_2) \in D} \left(
    1- \frac{\alpha_V}{k-1} \pi_{l_1}\pi_{l_2}\right)^{-1}
\]
\end{cor}

Having established a proper dynamic representation, Theorem~\ref{T:consistency} 
follows from Theorem~\ref{T:sameLaw} in complete analogy with the proof
of \cite[Thm.~5.1]{ArakSurg89} as in \cite[Thm.~1]{SchrLies10}. A discrete 
analogue of the Poisson line transect property \cite[Thm.~5.1.c]{ArakSurg89} holds
as well. Indeed, combining consistency with the boundary birth mechanism 
of the dynamical representation, we obtain the following.

\begin{cor} Let $l$ be a straight line that contains no nodes of ${\cal T}$. 
Then, the intersection points and intersection directions of $l$ with the 
edges of the polygonal field $\hat{\cal A}_{\Phi}$ coincide in distribution 
with the intersection points and directions of $l$ with the line field 
$\Lambda_{\cal T}$ defined to be the random sub-collection of ${\cal T}$ 
where each straight line $l^* \in {\cal T}$ is chosen to belong to 
$\Lambda_{\cal T}$ with probability ${\pi_{l^*}} / (1+\pi_{l^*})$ and 
rejected otherwise, and all these choices are made independently.
\end{cor}

To conclude this section, we turn to Markov properties that are
direct consequences of the dynamic representation. 

A spatial Markov property reminiscent of that enjoyed by the 
Arak--Clifford--Surgailis model in the continuum is the following.
For a piece wise smooth simple closed curve $\theta \subset {\mathbb R}^2$ 
containing no nodes of ${\cal T}$, the conditional distribution of 
$\hat{\cal A}_{\Phi}$ in the interior of $\theta$ depends on the configuration 
exterior to $\theta$ only through the intersections of $\theta$ with the edges 
of the polygonal field and through the colouring of the field along $\theta.$  

To relate our model to Gibbs fields commonly used in image analysis, assume 
for the remainder of this section that ${\cal{T}}$ forms a regular lattice and 
$D$ is an $m\times n$ rectangle. In this case, $D$ is divided by ${\cal{T}}$
in square cells, known as pixels. There is a one-to-one correspondence
between a coloured polygonal configuration $\hat \gamma$ and the array of pixel 
colours. Indeed, the colour of a pixel is that of the face of $\hat \gamma$ 
that it falls in, and, reversely, the edges of $\gamma$ are composed of the 
segments between pixels of different colours. In this dual framework, we
obtain the following local Markov factorisation.

\begin{cor}
Let $D$ be an $m\times n$ array and let ${\cal{T}} = \{ l_1, \dots, l_{m+n-2} \}$ 
be the corresponding regular linear tessellation with the indices chosen in such 
a way that for $i=1, \dots, n-1$,  $l_i$ is the horizontal line between the $i$-th 
and $(i+1)$-st row and for $i = 1, \dots, m-1$, $l_{n+i-1}$ is the vertical line
between the $i$-th and $(i+1)$-st column. Then the random element $\hat A_{\Phi_D}$ 
is the dual of a random vector $X = (X_1, \dots, X_{mn})$ of pixel values indexed
in column major order with a joint probability mass function $\PP(\hat \gamma) = 
\PP( X_1 = x_1; \cdots; X_{mn} = x_{mn} )$ that factorises as
\[
\PP( X_1 = x_1 ) \prod_{i=2}^n \PP( X_i = x_i \mid X_{i-1} = x_{i-1} ) 
 \prod_{i=1}^{m-1} \PP( X_{in+1} = x_{in+1} \mid X_{(i-1)n+1} = x_{(i-1)n+1} ) 
\]
\[
\times 
\prod_{i=1}^{m-1} \prod_{j=2}^n \PP( X_{in+j} = x_{in+j} \mid 
X_{(i-1)n+j-1} = x_{(i-1)n+j-1} ;
\]
\[
\quad \quad \quad \quad  X_{(i-1)n+j} = x_{(i-1)n+j} ; 
X_{in+j-1} = x_{in+j-1} )
\]
for $x_i \in \{ 1, \dots, k \}$, $i=1, \dots, mn$.
\end{cor}

\begin{proof} 
Choose the time direction in the dynamic representation in such a way that
the chronological order of the nodes coincides with the column major order and
recall that the conditional behaviour at a node depends only on the colours
and trajectories immediately to its `left', i.e.\ immediately prior to it in 
time. For the first pixel, $\PP( X_1 = x_1 )$ is the probability that the
initial colour is $x_1$, which is $1/k$ as this colour is chosen uniformly
from the set $\{ 1, \dots, k \}$. Next, the probabilities of the pixel values
in the first column are derived from the boundary birth mechanism. Thus, 
for $x_i, x_{i-1} \in \{ 1, \dots, k \}$ and $i=2, \dots, n$,
\[
\PP( X_i = x_i \mid X_{i-1} = x_{i-1} ) =
\left\{ \begin{array}{ll} 
\frac{1}{k-1} \, \frac{ \pi_{l_{i-1}} }
{ 1 + \pi_{l_{i-1}} } & \mbox{ if } x_i \neq x_{i-1} ; \\
\frac{ 1 }
{ 1 + \pi_{l_{i-1}} } & \mbox{ if } x_i = x_{i-1} .
\end{array}
\right. 
\]
Similarly, for $x_{in+1}, x_{(i-1)n+1} \in \{ 1, \dots, k \}$ and $i= 1, \dots, m-1$,
\[
\PP( X_{in+1} = x_{in+1} \mid X_{(i-1)n+1} = x_{(i-1)n+1} ) =
\left\{ \begin{array}{ll} 
\frac{1}{k-1} \,\frac{ \pi_{l_{i+n-1}} }
{ 1 + \pi_{l_{n+i-1}} } & \mbox{ if } x_{in+1} \neq x_{(i-1)n+1}; \\
\frac{ 1 }
{ 1 + \pi_{l_{n+i-1}} } & \mbox{ if } x_{in+1} = x_{(i-1)n+1}. 
\end{array}
\right.
\]

Use the shorthand notation $\PP_{ij}(x \mid u, v, w)$ for
\[
\PP( X_{in+j} = x_{in+j} \mid X_{(i-1)n+j-1} = u ; X_{(i-1)n+j} = v; X_{in+j-1} = w ).
\]
Then, for $u\neq v\neq w\neq u$,
\[
\PP_{ij}(x \mid u, v, v) = \left\{ \begin{array}{ll} 
\alpha_V & \mbox{ if } x = v \\
\alpha_X / (k-1) & \mbox{ if } x \neq v
\end{array}
\right.
\]
by {\bf (E3a)}, and
\[
\PP_{ij}(x \mid u, v, w) = \left\{ \begin{array}{ll} 
\alpha_T & \mbox{ if } x \in \{ v, w \} \\
\alpha_X / (k-1) & \mbox{ if } x \not \in \{ v, w \}
\end{array}
\right.
\]
by {\bf (E3b)} regardless of $i,j$. Furthermore, the
interior birth mechanism implies that
\[
\PP_{ij}(x \mid u, u, u) = \left\{ \begin{array}{ll} 
\alpha_V \pi_{l_{j-1}} \pi_{l_{n+i-1} } / (k-1)^2 
& \mbox{ if } x \neq u \\
1 -  \alpha_V \pi_{l_{j-1}} \pi_{l_{n+i-1} } / ( k-1 ) 
& \mbox{ if } x = u 
\end{array}
\right.
\]
and finally {\bf (E4)} determines the probabilities
\[
\PP_{ij}(x \mid u, u, v) = \left\{ \begin{array}{ll} 
1 - \epsilon \pi_{l_{j-1}} & \mbox{ if } x = v \\
\alpha_V \pi_{l_{j-1}} / (k-1) & \mbox{ if } x = u \\
\alpha_T \pi_{l_{j-1}} / (k-1) & \mbox{ if } x \not \in \{ u, v \}
\end{array}
\right.
\]
and
\[
\PP_{ij}(x \mid u, v, u) = \left\{ \begin{array}{ll} 
1 - \epsilon \pi_{l_{n+i-1}} & \mbox{ if } x = v; \\
\alpha_V \pi_{l_{n+i-1}} / (k-1) & \mbox{ if } x = u; \\
\alpha_T \pi_{l_{n+i-1}} / (k-1) & \mbox{ if } x \not \in \{ u, v \} .
\end{array}
\right. 
\]
Collecting all terms completes the proof.
\end{proof}

Models for which the above factorisation holds have been dubbed
mutually compatible Gibbs random fields by Goutsias \cite{Gout89}.
In particular, the interior birth mechanism, as well as the collisions 
and path propagation described in the dynamical representation's 
{\bf (E3)--(E4)} correspond to the local transfer function in 
\cite{Gout89}, see also\cite{ChamIdieGous98}.

Note that we have chosen the time direction so as to conform to 
column major order. A fortiori, such models are Markov random
fields with the second order neighbourhood structure in which horizontally,
vertically, and diagonally adjacent pixels are neighbours. A similar 
factorisation holds for any choice of the time axis. Moreover, there 
is no need for $D$ to be a rectangle. Indeed, because of the assumptions 
on ${\cal{T}}$, every interior node is hit by exactly two segments that
are adjacent to three, not necessarily rectangular, pixels. See 
Figure~\ref{F:ukraine} for an example. The notation, however, becomes 
more cumbersome. For this reason, we prefer to use the graph-theoretical 
formulation with its neater formulae.


\section{Birth and death dynamics}
\label{S:sampler}

In this section, we use the dynamic representation of 
Section~\ref{S:dynamics} to propose dynamics that are reversible and 
leave the law of $\hat{\cal{A}}_{\Phi_D}$ invariant. These dynamics 
will serve as stepping stone for building Metropolis--Hastings dynamics 
for the general polygonal field models (\ref{PolyField}). For the 
two-colour case, algorithms inspired by dynamic representations can 
be found in \cite{KlusLiesSchr05,Schr05,SchrLies10,MatuSchr12}. In
that case, however, it is sufficient to focus on colour blind polygonal
configurations as the colouring is completely determined by the colour 
at a single point. For $k>2$, this is no longer the case and we have to 
explicitly incorporate colours in our dynamics. Moreover, particles do 
not necessarily die upon collision and the disagreement loop principle 
of Schreiber \cite{Schr05} no longer applies.

The basic continuous time dynamics we propose consist of adding and deleting 
particle birth sites. In order to fully explore the state space, recolouring 
will also be necessary, at some fixed rate $\tau>0$. We work with a constant 
death rate $1$. The birth rate at a boundary entry point ${\iin} (l, D)$ 
and a vacant internal node $n(l_1, l_2)$ are set to $\pi_l$ and 
\(
\alpha_V \pi_{l_1} \pi_{l_2} / ( k - 1 - \alpha_V \pi_{l_1} \pi_{l_2} )
\)
to satisfy the detailed balance equations
\[
\frac{\pi_l}{1+\pi_l} \times 1 = \left( 1 - \frac{\pi_l}{1+\pi_l} \right)
\times \mbox{birth rate}( {\iin} (l,D) )
\]
respectively
\[
\frac{\alpha_V}{k-1} \pi_{l_1} \pi_{l_2} \times 1 =  
\left( 1 - \frac{\alpha_V}{k-1} \pi_{l_1} \pi_{l_2} \right ) \times
\mbox{birth rate}( n(l_1, l_2) ).
\]
Recall that if $n(l_1, l_2)$  is hit by some previously born particle, the 
birth is discarded. 
For computational convenience, we shall keep track of the discarded
births during the dynamics.

In case of a birth update, the particle(s) emitted by the birth site are given 
trajectories in accordance with {\bf E4}. Upon collisions, {\bf E3} is 
invoked. Whenever possible, existing trajectories are re-used. A dual 
reasoning is applied to deaths. 

To make the above ideas precise, suppose that the current state is
$\hat \gamma$, understood here to {\em include\/} the knowledge of all discarded 
birth sites, which we modify by adding or deleting a (discarded) birth site 
to obtain $\hat\gamma^\prime$. We shall use the following segment classification:
\begin{description}
\item[plus] the segment does not lie on any edge of $\gamma$ but it does
lie on some edge of $\gamma^\prime$;
\item[minus] the segment lies on some edge of $\gamma$ but it does not
lie on any edge of $\gamma^\prime$;
\item[changed] the segment lies on some common edge of $\gamma$ and
$\gamma^\prime$ but the colour of at least one of its adjacent 
polygonal faces has changed, or the segment lies in the interior of
faces of $\hat \gamma$ and $\hat \gamma^\prime$ having different colours.
\end{description}

The dynamics are now as follows. In case of a birth at node $n(l_1, l_2)$ with
coordinates $(t,y)$, two plus segments arise along $l_1$ and $l_2$ forward
in time. In case of a boundary birth, a single plus segment is generated.
Similarly, in case of a death, one or two minus segments arise. We order the nodes
with first coordinate larger than $t$ in chronological order and update them
one at a time until some further time $t'>t$  for which the intersections
of $\hat \gamma$ and $\hat \gamma^\prime$  with the vertical line specified
by first coordinate $t'$ are identical. 

At each node, we first check whether the node is hit either by some segment marked
`plus' or `minus', or by some `changed' segment of $\gamma$. When using the word `hit' 
we shall always mean that the tail of the segment has a smaller time coordinate 
than the node at its head. We shall use the phrase `emanating segments' for those segments
whose tail is at the node and whose head has a larger time coordinate than 
the node. If no such segment exists, we need to check whether the node is a
non-discarded birth site. If it is, due to e.g.\ nesting, the colour just prior 
to the birth site may be different in $\hat \gamma$ and $\hat \gamma^\prime$. In 
this case, a new colour is chosen for the region to the right of the node
from those not equal to the colour just prior to the node. Otherwise, the
status quo is propagated. 

If the node is hit by two marked segments (`plus', `minus' or `changed' and
belonging to $\gamma$), or by one such segment and one that is an unchanged
common edge of $\gamma$ and $\gamma^\prime$ (which we label `old') a collision 
update is made as outlined in Table~\ref{T:collision}. 

\begin{table}
\begin{tabular}{|l|l|}
\hline
minus/minus & label emanating segments; \\
at vacant node & \\
\hline
minus/minus & implement birth by choosing new colour from those \\
at discarded birth site &  not equal to that prior to the discarded birth site; \\
& label emanating segments; \\
\hline
minus/old & invoke E4; \\
minus/plus & label emanating segments; \\
minus/changed & \\
\hline
plus/old & invoke E3; \\
plus/changed & label emanating segments; \\
old/changed & \\
plus/plus at & \\
\mbox{ } vacant node & \\
\hline
plus/plus & discard the birth;\\
at birth site & invoke E3 and label emanating segments; \\
\hline
changed/changed & check whether the colours above and below the face \\
& bounded by the two hitting segments agree 
 in $\hat \gamma$ and $\hat \gamma^\prime$; \\
& if so, do nothing; \\
& otherwise invoke E3 and label emanating segments. \\
\hline
\end{tabular}
\caption{Collision updates.}
\label{T:collision}
\end{table}
 
In the remaining case that the node is hit by a single segment marked either
`plus' or `minus' or by a segment of $\gamma$ labelled `changed', the path is 
updated as outlined in Table~\ref{T:path}.

\begin{table}
\begin{tabular}{|l|l|}
\hline
plus path & discard the birth; \\
at birth site & invoke E4; \\
& label emanating segments; \\
\hline
plus path & invoke E4; \\
at vacant site & label emanating segments; \\
\hline
minus path & implement birth by choosing new colour from those \\
at discarded birth site &  not equal to that prior to the discarded birth site; \\
& label emanating segments; \\
\hline
minus path & label emanating segments; \\
at vacant site & \\
\hline
changed path & in case the path splits, choose a new colour \\
& from those not equal to the colours above and below the path;  \\
& label emanating segments. \\
\hline
\end{tabular}
\caption{Path updates.}
\label{T:path}
\end{table}

An illustration is given in Figure~\ref{F:update}. The current node $n$ to 
be updated is in the middle of the panels. In $\hat \gamma$, the node is a 
birth site: no segments hit $n$ and there are two emanating segments. In the 
new polygonal configuration $\hat \gamma^\prime$, the node is hit by a segment 
separating the green from the blue face which is therefore labelled `plus'. 
According to Table~\ref{T:path}, in $\hat \gamma^\prime$, the birth gets 
discarded and we invoke {\bf E4}, say resulting in the decision to split. 
Hence, we must also choose a colour other than green or blue for the region to
the right, that is, forwards in time, for example red. The labels of the two 
emanating segments are `changed' for the one separating the blue and red faces, 
and `old' for the one forming the boundary between the red and green faces.

\begin{figure}[hbt]
\begin{center}
 \includegraphics[height=2.5cm,width=2.5cm]{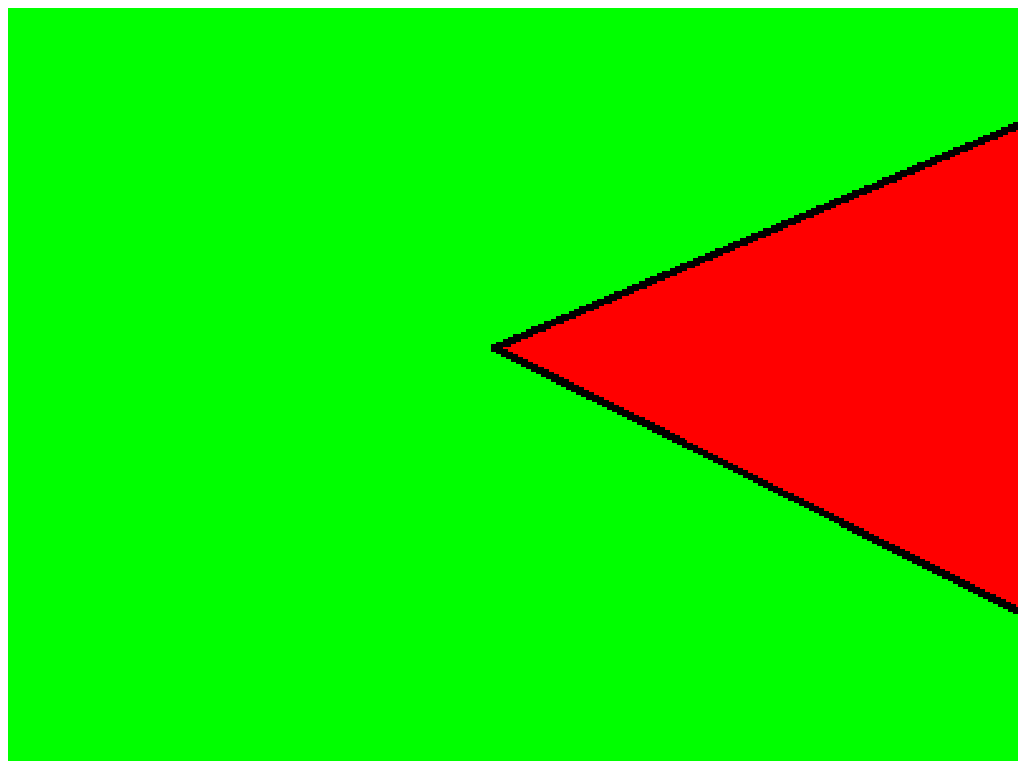}
 \includegraphics[height=2.5cm,width=2.5cm]{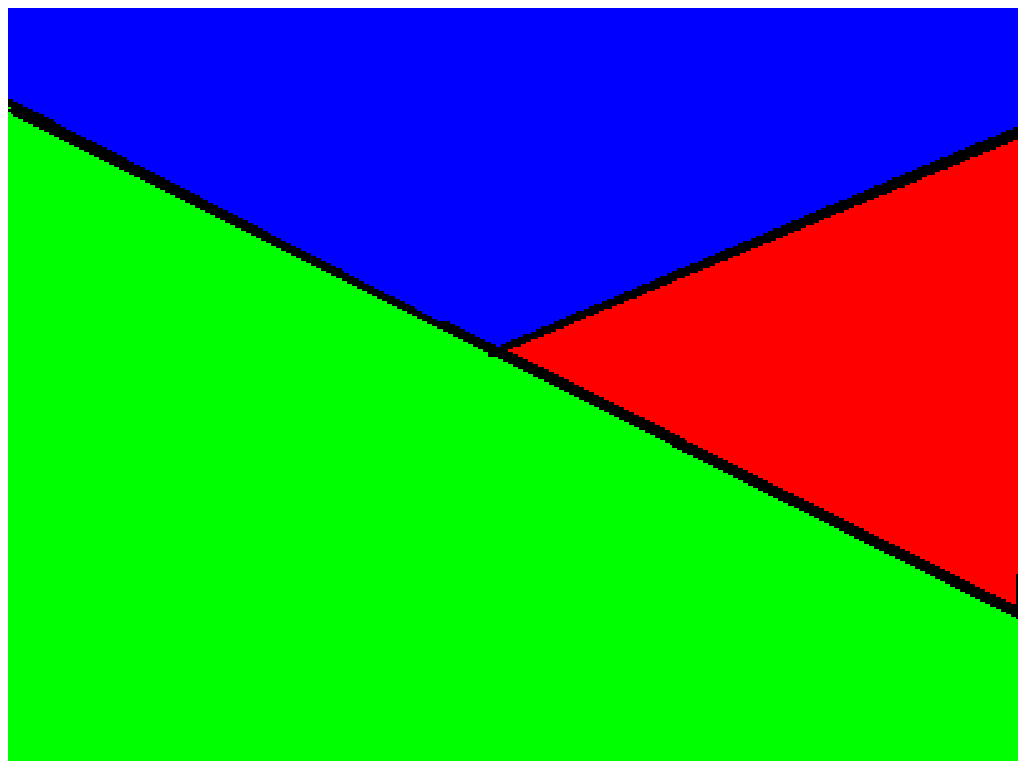}
\end{center}
\caption{Left: $\hat \gamma$. Right: $\hat \gamma^\prime$.}
\label{F:update}
\end{figure}

For recolouring, classic local colour switches are used, as detailed
in for example \cite{Wink03}. As for the two-colour case in \cite{SchrLies10}, 
we obtain the following result.

\begin{thm} \label{T:invariant}
The distribution of the consistent polygonal field $\hat{\cal{A}}_{\Phi_D}$
is the unique invariant probability distribution of the birth-death-recolour 
dynamics described above upon ignoring discarded birth attempts, to which 
they converge in total variation from any initial state $\hat \gamma \in 
\hat\Gamma_D({\cal{T}})$ for which $\PP( \hat{\cal{A}}_{\Phi_D} = 
\hat \gamma ) > 0$.
\end{thm}

\begin{proof}
The case $k=2$ was considered in \cite{SchrLies10}. Hence assume $k>2$.
The total transition rate is bounded from above by a positive constant.
Indeed, for each internal birth site $n(l_1,l_2)$, the rate
\[
\frac{ \alpha_V \pi_{l_1} \pi_{l_2} / (k-1) }{
1 -  \alpha_V \pi_{l_1} \pi_{l_2} / (k-1) } \leq \frac{ 1 / (k-1) }{ 1 - 1/(k-1) }
= \frac{1}{k-2} \leq 1.
\]
The death rate for a (discarded) birth at $n(l_1, l_2)$ is equal to $1$. 
Similarly, the birth rate $\pi_l$ at entry points $\iin(l, D) = \pi_l \leq 1$
and the death rate is again equal to $1$. Therefore, an upper bound is the
sum of $\tau$, the number of nodes and the number of lines hitting $D$. Thus,
our dynamics can be algorithmically generated by a Poisson clock of constant
rate, and an embedded Markov transition matrix that governs the transitions.
This transition matrix restricted to admissible coloured polygonal configurations
with positive probability under the putative invariant probability distribution
is irreducible, since any state can be reached from any other state by successively 
removing all birth attempts, choosing an appropriate initial colour and then 
building the target state by successively adding particles. Hence the dynamics 
constitute a finite state space irreducible Markov process and there exists a 
unique invariant probability distribution. See for example Theorem~20.1 in 
\cite{Levietal}. The same theorem also yields the converge in total variation. 
The invariance of the distribution including discarded birth attempts follows 
from the invariance of the Bernoulli birth site probabilities under the dynamics, 
the fact that the trajectories preserve {\bf E1--E4} by design, and the well-known 
invariance of the local colour switches \cite[Section~10.2]{Wink03}. Summing
over any discarded birth attempts completes the proof.
\end{proof}

In fact, the dynamics are reversible. Consider for example the collision 
of a positive path with a birth site as in Figure~\ref{F:update}. 
In this example, the transition rate is multiplied by 
$(k-2) \alpha_T \pi_{l_j} /(k-1)$ for the split and $1/(k-2)$ for the colour,
amounting to $\alpha_T \pi_{l_j} / (k-1)$ whereas the probability of 
$\hat \gamma$ {\it including the knowledge of discarded births} gains a 
factor $(k-1) \alpha_T \pi_{l_j} / (k-1) = \alpha_T \pi_{l_j}$ (the first
term $k-1$ to compensate for the choice of colour in $\hat \gamma$, the 
second for the split and colour in $\hat \gamma^\prime$). The reverse 
collision of a minus path with a discarded birth site yields a factor 
$1/ (k-1)$ for the rate. Thus, this type of collision is reversible.  
Similar calculations can be made for the other types of collision and are 
left to the reader. 

For general polygonal field models (\ref{PolyField}) with a Hamiltonian that
is the sum of (\ref{E:consistent}) and some other term ${\cal{H}}_D$,
we propose a Metropolis--Hastings algorithm. The algorithm has the same birth, 
death and recolour rates as the dynamics presented in Section~\ref{S:sampler}.
The difference is that a new state $\hat \gamma^\prime$ is accepted with 
probability
\begin{equation}
\label{e:MH}
\min \left( 1,
\exp\left[ 
   {\cal{H}}_D(\hat \gamma) - {\cal{H}}_D(\hat \gamma^\prime)
\right] \right),
\end{equation}
whereas the old state is kept with the complementary probability. An 
example of ${\cal{H}}_D$ will be presented in Section~\ref{S:application}.

\begin{thm} \label{T:MH}
Let ${\cal{H}}_D: \hat{\Gamma}_D({\cal T}) \mapsto {\mathbb R}$ be finite. Then 
the distribution of the polygonal field $\hat{\cal{A}}_{\Phi_D + {\cal{H}}_D}$ is 
the unique invariant probability distribution of the Metropolis--Hastings 
dynamics. Moreover, these dynamics converge in total variation to 
$\hat{\cal{A}}_{\Phi_D + {\cal{H}}_D}$  from any initial state 
$\hat \gamma \in \hat\Gamma_D({\cal{T}})$ for which
$\PP( \hat{\cal{A}}_{\Phi_D} = \hat \gamma ) > 0$.
\end{thm}

\begin{proof}
By the assumption on the Hamiltonian and arguments analogous to those in the proof
of Theorem~\ref{T:invariant}, the embedded Markov transition matrix that governs the 
transitions is irreducible. Hence the dynamics constitute a finite state space 
irreducible Markov process and there exists a unique invariant probability distribution, 
cf.\ Theorem~20.1 in \cite{Levietal}, to which they converge in total variation. By 
Theorem~\ref{T:invariant}, the birth-death-recolour dynamics leave the distribution of 
$\hat{\cal{A}}_{\Phi_D}$ invariant. The modification by the acceptance probabilities 
(\ref{e:MH}) yields that the target distribution $\hat{\cal{A}}_{\Phi_D + {\cal{H}}}$ is 
left invariant by the Metropolis--Hastings dynamics. 
\end{proof}

\section{Application to linear network extraction}
\label{S:application}

The goal of this section is to apply the model presented in Section~\ref{S:definition} 
to the extraction of a network of tracks in between crop fields from image data. The 
left-most panel in Figure~\ref{F:ukraine}, obtained from the collection of publicly 
released SAR (Synthetic Aperture Radar) images at the NASA/JPL web site
{\tt http://southport.jpl.nasa.gov} shows an agricultural region in Ukraine. A pattern 
of fields separated by tracks is visible, broken by some hamlets. The image was previously 
analysed by Stoica et al.\ \cite{Stoietal02} by means of a Markov line segment process. 

\begin{figure}[htb]
\begin{center}
  \fbox{\includegraphics[height=65mm]{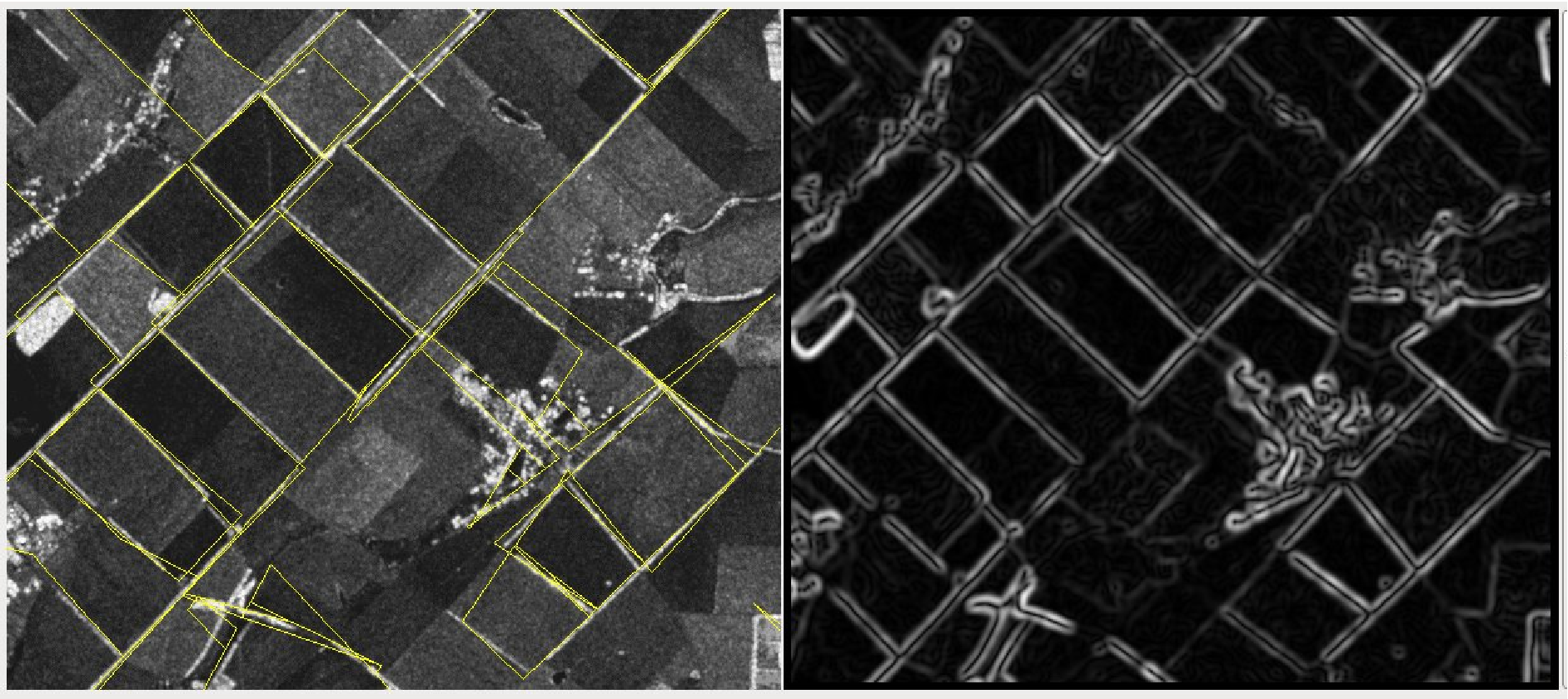}}

  \fbox{\includegraphics[height=65mm]{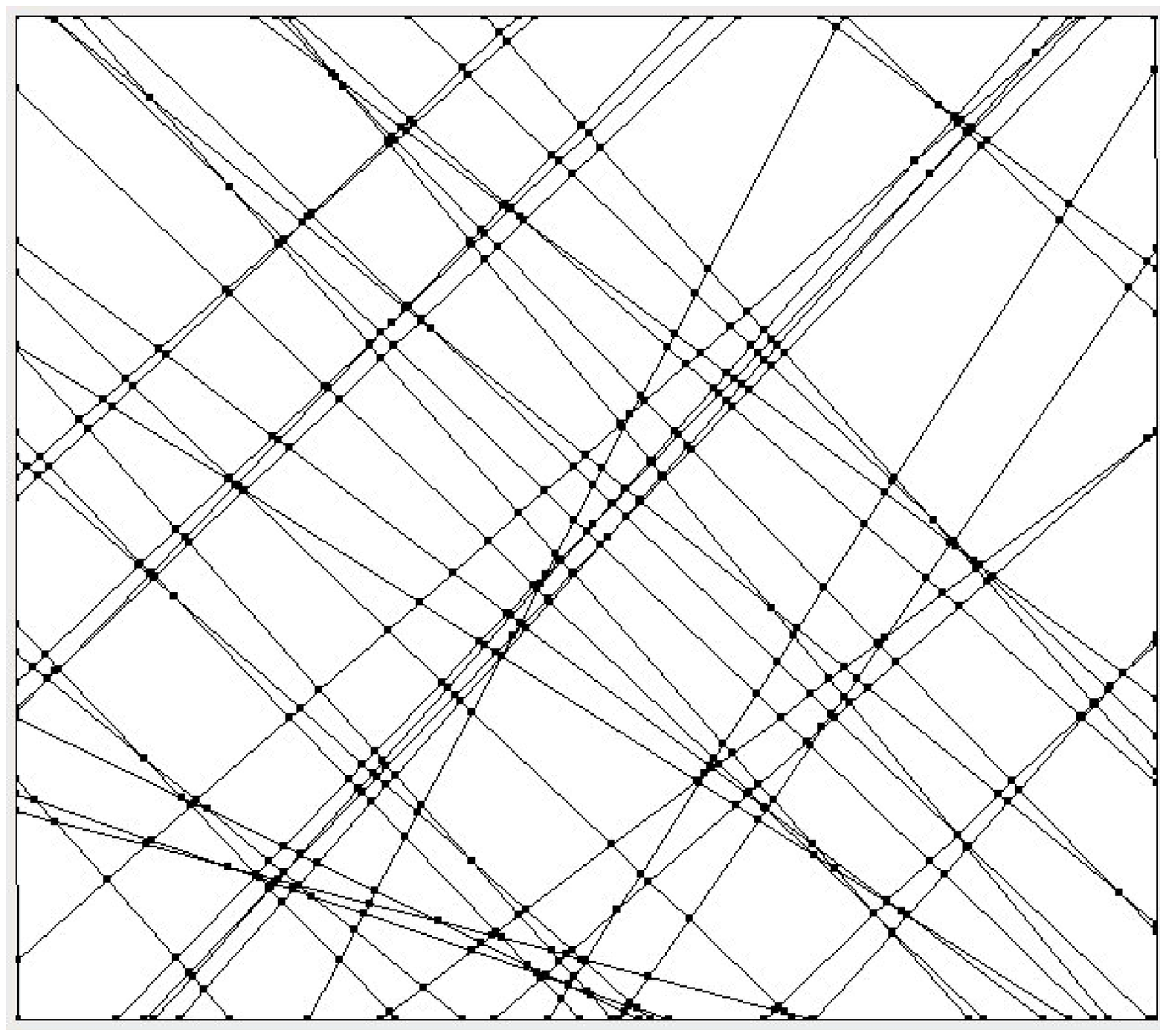}}
\end{center}
\caption{Polygonal configuration (in yellow) overlaid on a SAR image of 
fields in rural Ukraine (top left panel), the corresponding edge map (top 
right panel) and a regular linear tessellation extracted from the Hough
accumulation array (bottom row).}
\label{F:ukraine}
\end{figure}

Note that the tracks that run between adjacent fields show up in the image
as whitish lines against the darker fields. Thus, a track is associated with
a high image gradient. To find a suitable family ${\cal{T}}$ of straight lines, 
we therefore begin by computing the gradient of the data image after 
convolution with a radially symmetric Gaussian kernel with standard deviation 
$\sigma = 3$ to suppress noise. The right-most panel in Figure~\ref{F:ukraine} 
shows the gradient length thus obtained. We then compute the Hough transform 
\cite{DudaHart72,Houg62} in an $80\times 80$ accumulation array and select the 
$8$ lines corresponding to the bins with the largest number of accumulated votes. 
This collection is augmented by lines parametrised by the largest local extrema 
in the accumulation array to yield the $42$ lines shown in the bottom panel of 
Figure~\ref{F:ukraine}. The resulting set  is finite and contains no three lines that 
intersect in a common point, in other words, it is a regular linear tessellation. 
Regarding the choice of $\alpha_V$, a comparison of the data with the simulations 
shown in Figure~\ref{F:Arak} suggests $\alpha_V = 1/2$. Since the lines are
selected on the basis of a high gradient value, there seems to be no particular 
reason to favour one line over another, so we may set the line activity to a 
constant value.

To quantify how well a polygonal configuration fits the data, recall that
an edge should be present when there is a large gradient, and absent when the
gradient is small. This desirable property is captured by the Hamiltonian
\begin{equation}
\label{E:flux}
{\cal{H}}_D(\hat \gamma) = - \beta \sum_{ e \in E(\gamma) } \left[
f(e) - c(e) \right],
\end{equation}
where $f(e)$ is the integrated absolute gradient flux along edge $e$, $c(e)$ a 
threshold to discourage spurious edges, and $\beta > 0$. We take $c(e)$ proportional
to the number of segments along the edge with proportionality constant $c>0$. For
this choice, (\ref{E:flux}), being a sum of segment contributions, is local in 
nature, which is convenient from a computational perspective.

To find an optimal polygonal configuration, we use simulated annealing 
applied to the Metropolis--Hastings algorithm of Section~\ref{S:dynamics} for
$\hat{\cal{A}}_{\Phi_D + {\cal{H}}_D}$  with $\Phi_D$ given by (\ref{E:consistent})
and ${\cal{H}}_D$ defined by (\ref{E:flux}) and recolour rate $100$. Starting at 
temperature $1/\beta = 10$, the inverse temperature parameter $\beta$ is slowly 
increased to $100$ according to a geometric cooling schedule. 
The result for $c=2$, line activity $1/2$ and $k=4$ is shown in Figure~\ref{F:ukraine}.
There are no false negatives; a few false positives occur near the hamlets and are 
connected to the track network. The precision of the line placement is clearly linked 
to the precision of the underlying regular linear tessellation.

\section{Summary and discussion}

In this paper, a new class of consistent random fields was introduced whose 
realisations are coloured mosaics with not necessarily convex polygonal 
tiles. The vertices of the polygonal tiles may have degrees two, three or
four. The construction, inspired by the Arak--Clifford--Surgailis polygonal 
Markov fields in the continuum \cite{ArakSurg89,ArakSurg91,ArakClifSurg93}, 
extends our previous construction \cite{SchrLies10} of consistent polygonal 
random fields with tile vertices of degree two only. The latter case is 
substantially simpler due to the fact that interactions are restricted to a 
hard core constraint only, and, moreover, for a given realisation, there are
only two equally likely admissible colourings. We developed a dynamic representation 
for consistent multi-colour polygonal fields, which was used to prove the basic 
properties of the model including consistency and an explicit expression for the 
partition function. Local and spatial Markov properties were also considered. 
The dynamic representation provided the foundation on which to build 
Metropolis--Hastings style samplers for Gibbsian modifications of these fields. 
Finally, we applied the model to the detection of linear networks in rural scenes. 
A modification of our models would consist in ascribing activity parameters
to segments rather than lines, cf.\ \cite{MatuSchr12}. A disadvantage seems
to be that the dynamic representation would depend on the direction of time,
in other words, the model would be anisotropic.

To conclude, we should stress that, although the model was inspired by 
those of Arak et al.\,~there are striking differences inherent to the 
discrete set-up. Notably, in the continuum collinear edges are not allowed,
in the discrete set-up they are. Indeed, if one were to forbid collinear 
edges, this would lead to a forbidden line whose influence would be felt 
at arbitrarily large distance from its single edge, hence ruling out any 
meaningful Markovianity. This is not true in the continuum as the dynamic
representation there ensures that collinear edges occur with probability zero.
As a consequence, our consistent random fields are not Arak--Surgailis fields
conditional on having their edges along the lines of ${\cal{T}}$. More
fruitfully, our models provide a graph-theoretical interpretation of mutually
compatible Gibbs random fields that inspires novel simulation algorithms as 
an alternative to the usual local tile updating schemes \cite{Wink03}.

\section*{Acknowledgements}
This research was supported by The Netherlands Organisation for Scientific
Research NWO (613.000.809). The author is grateful to T. Schreiber for a
pleasant and interesting collaboration on polygonal Markov fields and their 
applications that was cut short by his untimely death, to H.\ Noot for 
programming assistance, to K.\ Kayabol for reading a preliminary draft and
to the anonymous referees for their careful reading and helpful suggestions.

\end{document}